\documentclass[12pt]{iopart}
 
\usepackage{mathrsfs}\usepackage{color}
\usepackage{amsfonts}
\usepackage{amsthm,amssymb}
\usepackage{graphicx} 
\usepackage{lscape}
\usepackage{subeqn} 
\usepackage{iopams,slashbox}  
\newcommand{\eps}{\epsilon}
\newcommand{\lmm}{\lim_{t\rightarrow-\infty}}
\newcommand{\lmp}{\lim_{t\rightarrow\infty}}

\newcommand{\lmt}{\lim_{t\rightarrow t_M^-}}
\newcommand{\beq}{\begin{equation}}
\newcommand{\eeq}{\end{equation}}

\newcommand{\bea}{\begin{eqnarray}}
\newcommand{\eea}{\end{eqnarray}}
\newcommand{\beas}{\begin{eqnarray*}}
\newcommand{\eeas}{\end{eqnarray*}}
\newtheorem{theorem}{Theorem}[section]
\newtheorem{proposition}{Proposition}[section]
\newtheorem{lemma}{Lemma}[section]

\newcommand{\LL}{|\lambda|}
\newcommand{\refb}[1]{(\ref{#1})}

\newcommand{\Lag}{\mathscr{L}}
\newcommand{\bsy}[1]{\boldsymbol{#1}}

\begin{document} 

\title[Collapse of a cylindrical scalar field with non-minimal coupling]{Collapse of a self-similar cylindrical scalar field with non-minimal coupling II: strong cosmic censorship}
\author{Eoin Condron and Brien C. Nolan}

\address{School of Mathematical Sciences, Dublin City University, Glasnevin, Dublin 9, Ireland.}
\eads{ \mailto{eoin.condron4@mail.dcu.ie}, \mailto{brien.nolan@dcu.ie}}

\begin{abstract}
We investigate self-similar scalar field solutions to the Einstein equations in whole cylinder symmetry. Imposing self-similarity on the spacetime gives rise to a set of single variable functions describing the metric. Furthermore, it is shown that the scalar field is dependent on a single unknown function of the same variable and that the scalar field potential has exponential form. The Einstein equations then take the form of a set of ODEs. Self-similarity also gives rise to a singularity at the scaling origin. We extend the work of \cite{Part1}, which determined the global structure of all solutions with a regular axis in the causal past of the singularity. We identified a class of solutions that evolves through the past null cone of the singularity. We give the global structure of these solutions and show that the singularity is censored in all cases.  
\end{abstract}

\section{Introduction \& Summary}
This is the second of two papers which aim to give a rigorous analysis of self-similar cylindrical spacetimes coupled to a non-linear scalar field. In particular, we are interested in determining whether a subset of these spacetimes exhibit naked singularity formation. In \cite{Part1}, it was shown that the assumption of self-similarity of the first kind \cite{Carr}, where the homothetic vector field is assumed to be orthogonal to the cylinders of symmetry, gives rise to a singularity at the scaling origin $\mathcal{O}$ (the point at which the homothetic Killing vector is identically zero). This point lies on the axis of symmetry. Solutions emanating from a regular axis to the past of $\mathcal{O}$ were studied and the global structure of solutions was given in the region bounded by the axis and the past null cone $\mathcal{N}_-$ of the singularity, which we call region \textbf{I}. The system has two free parameters labelled $V_0$ and $k$, and the global structure was given for all possible values of the parameters. The assumptions reduce the coupled Einstein field equations to a set of ODEs, and these naturally give rise to an initial value problem with data on the regular axis. There is also a free initial datum, $l_0$, on the regular axis. The independent variable $\eta$ is a similarity variable normalised so that $\eta=1$ on the regular axis and $\eta=0$ on the past null cone $\cal{N}_-$ of $\cal{O}$. \\
It was shown that for $(k^2,V_0,l_0)\in \bar{K},$ where 
\beq
\nonumber\bar{K}=\{(k^2,V_0,l_0):k^2\ge 2\}\cup\{(k^2,V_0,l_0):k^2<2,V_0e^{(k^2/2-1)l_0}\ge k^2/8\},
\eeq
the solutions terminate on or before $\cal{N}_-$. Specifically, there is a value $\eta_M\in[0,1)$ such that the hypersurface at $\eta=\eta_M$ corresponds either to future null infinity (see cases 1 and 2 of Fig. 1) or to a spacetime singularity (see cases 4 and 5 of Fig. 1). \\

\begin{figure}
\includegraphics[width=6.5in]{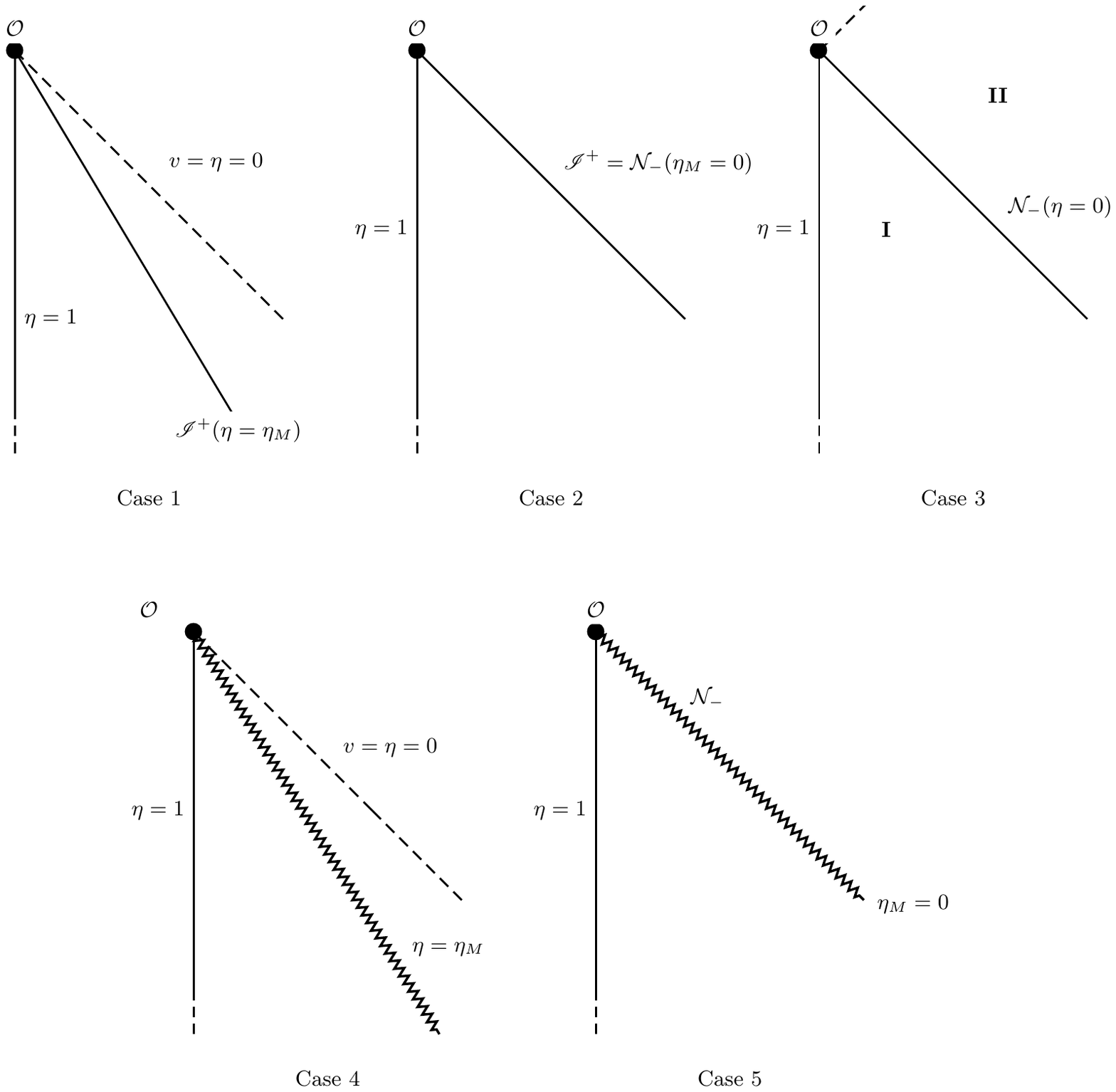}
\vskip -250pt
\caption{Possible structures of the spacetime. Case 3 depicts the spacetimes corresponding to values $(k^2,V_0,l_0)\in K)$, which are the subject of this paper. The remaining cases are various other subcases. $u,v$ are respectively retarded and advanced null coordinates and $\eta=v/u$.}
\end{figure}

\noindent 
We note that the spacetimes which have a singularity at $\eta=\eta_M\in[0,1)$ are singular at all times: there is no spacelike slice $\Sigma$ which avoids the singularity. Thus there is no spacelike slice along which we can impose initial data for the Einstein equations, and so this class of spacetimes is not relevant to the issue of cosmic censorship. \\
For $(k^2,V_0,l_0)\in K$, where $K$ is the complement in $(0,+\infty)\times\mathbb{R}^2$ of $\bar{K}$, it was shown that $\mathcal{N}_-$ is a regular surface that exists as part of the spacetime and the solutions may be extended into the region beyond $\mathcal{N}_-$. 
This is Case 3 in Figure 1. We define region \textbf{II} as the region bounded by $\mathcal{N}_-$ and the (putative) future null cone of the origin, $\mathcal{N}_+$. Our aim is to obtain the global structure of these solutions in this region and determine whether $\mathcal{N}_+$ exists as part of the spacetime. In other words, we seek to determine whether or not the singularity $\cal{O}$ is naked. 
In Section 2 we give a summary of the formulation of the field equations from \cite{Part1} and cast them as a dynamical system in a new set of variables. In Section 3 we give the asymptotic behaviour of solutions at $\mathcal{N}_-$, which is a fixed point of the dynamical system, and corresponds to the limit $t\rightarrow-\infty$, where $t$ is the independent variable. 
Section 4 contains an analysis of the remaining fixed points which are possible end states of solutions which reach the surface ${\cal N}_+$. We then determine the global behaviour of solutions in Section 5 and show that, for all solutions, the maximal interval of existence is bounded above. \\
The main result of the paper is established in Section 6. We quote the relevant theorem here:

\begin{theorem}
The class of spacetimes with line element \refb{LE eta}, subject to the Einstein-Scalar Field equations (10) with $(k^2,V_0,l_0)\in K$ and the regular axis conditions \refb{axis} satisfy strong cosmic censorship: the spacetimes are globally hyperbolic and $C^1$-inextendible. 
\end{theorem}

To prove this theorem, we present a number of results giving the global structure of the spacetimes, showing that the spacetimes are globally hyperbolic. To prove $C^1$-inextendibility, we show that a certain invariant of the spacetime, which depends only on the metric and its first derivatives, blows up at the spacelike singularity. Two cases arise; in the first case, this spacelike hypersurface corresponds to a scalar curvature singularity and in the second case it corresponds to a non-regular axis. $C^1$-inextendibility holds in both cases.\par
Before proceeding to the technicalities leading up to the proof of Theorem 1.1, we make some general comments.
Theorem 1.1, which builds on the results of Paper I, establishes that strong cosmic censorship holds for cylindrical spacetimes coupled to non-minimally coupled scalar fields in the case of self-similarity. Thus it provides a partial extension of the results of \cite{Berger}: we note that the non-minimally coupled scalar field does not satisfy the energy conditions required in \cite{Berger}.\par
Self-similarity forces the potential of the non-minimally coupled scalar field to assume an exponential form (see e.g.\ \cite{WE} and \cite{Carot-Colligne} for a detailed proof). In spherical symmetry, non-minimally coupled scalar fields have been considered in \cite{Malec} and \cite{Dafermos1}. Dafermos established that when the potential is bounded below by a constant (which can be negative), certain types of singularity are ruled out. Furthermore, weak cosmic censorship follows if the existence of a single trapped surface can be established \cite{Dafermos1}. In the present case, this condition on the potential corresponds to $V_0>0$ (see equation (8) below). However, our strong cosmic censorship result also holds when $V_0<0$. Thus it would be of interest to see if the results of the present paper extend to the spherically symmetric case, with and without the assumption of self-similarity. 

Scalar fields with an exponential potential have also been discussed extensively in the context of cosmology, where the role of the potential as a driver of inflation and accelerated expansion is of particular note. Homogeneous and isotropic models were first considered in \cite{halliwell} and there is now a significant body of literature on these models. Of particular note are the deep results on nonlinear stability, in the absence of symmetries, obtained in \cite{heinzle&rendall} and \cite{ringstrom}.

\section{Self-similar cylindrically symmetric spacetimes coupled to a non-linear scalar field}
We consider cylindrically symmetric spacetimes with whole-cylinder symmetry \cite{A&T} (see also \cite{cylindrical-collapse,HNN}). This class of spacetimes admits a pair of commuting, spatial Killing vectors $\bsy{\xi}_{(\theta)}$, $\bsy{\xi}_{(z)}$ called the axial and translational Killing vectors, respectively. Introducing double null coordinates $(u,v)$ on the Lorentzian 2-spaces orthogonal to the surfaces of cylindrical symmetry, the line element may be written as:
\begin{equation}
\label{LE uv}ds^2 = -2e^{2\bar{\gamma}+2\bar{\phi}}dudv + e^{2\bar{\phi}}r^2d\theta^2 + e^{-2\bar{\phi}}dz^2,
\end{equation}
where $r$ is the radius of cylinders, $\bar{\gamma} ,\bar{\phi}$  and  $r$  depend on $u$ and $v$ only.\\
We take the matter source to be a cylindrically symmetric, self-interacting scalar field $\psi(u,v)$ with stress-energy tensor given by
\begin{equation} 
\label{EMT}T_{ab}= \nabla_{a}\psi\nabla_{b}\psi-\frac{1}{2}g_{ab}\nabla^{c}\psi\nabla_{c}\psi-g_{ab}V(\psi),
\end{equation}
where $V(\psi)$ is the scalar field potential. The minimally coupled case $V\equiv 0$ was dealt with in \cite{Part1} and so we assume $V\neq0$. The line element is preserved by the coordinate transformations
\begin{equation}
\label{freedom}u \rightarrow \bar{u}(u),\quad v \rightarrow \bar{v}(v), \quad z \rightarrow \lambda z,
\end{equation}
for constant $\lambda$. Note that $\theta\in[0,2\pi)$ and so transformations of the kind $\theta\rightarrow \lambda \theta$ are not allowed in general. 
We assume self-similarity of the first kind \cite{Carr}, which is equivalent to the existence of a homothetic Killing vector field $\boldsymbol{\xi}$ such that
\begin{equation}
\label{SS}\mathcal{L}_{\boldsymbol{\xi}} g_{\mu \nu} = 2 g_{\mu \nu},
\end{equation}
where $\mathcal{L}_{\boldsymbol{\xi}}$ denotes the Lie derivative along the vector $\boldsymbol{\xi}$. We make the further assumption that $\boldsymbol{\xi}$ is cylindrical. The limitations of this assumption are discussed in \cite{Part1}. Equation \refb{SS} gives the form $\boldsymbol{\xi}=\alpha(u)\partial_u+\beta(v)\partial_v$ and the coordinate freedom \refb{freedom} is used to set $\alpha(u)=2u, \beta(v)=2v$.
Equations \refb{SS} then lead to 
\beq
\label{SSS}\bar{\gamma} =\gamma(\eta), \qquad\bar{\phi}=\phi(\eta)-\log|u|^{1/2},\qquad r =|u|S(\eta),
\eeq
where 
\beq\eta = \frac{v}{u}
\eeq
is called the similarity variable. The self-similar line element is then given by
\beq
\label{LE eta} ds^2 = -2|u|^{-1}e^{2\gamma(\eta)+ 2\phi(\eta)}dudv+|u|e^{2\phi(\eta)}S^2(\eta)d\theta^2+|u|e^{-2\phi(\eta)}dz^2.
\eeq
It was shown in [1] that in this coordinate system the self-similar, non-minimally coupled scalar field and its potential have the form
\beq
\psi = \frac{k}{2}\left(l(\eta) + \log|u||\eta|^{1/2}\right),\qquad V(\psi)=\frac{\bar{V}_0e^{-l(\eta)}}{|u||\eta|^{1/2}},
\eeq
for a function $l$ and constants $\bar{V}_0\neq0,k\neq0$. The field equations then reduce to (see \cite{Part1})
\begin{subequations}\bea
\label{FE a}2\gamma+2\phi=\frac{k^2l}{2}-\frac{1}{2}\log|\eta|+c_1,\\
\label{FE b}\eta S''=-V_0|\eta|^{-1}e^{\lambda l}S,\\
\label{FE c}2S'\gamma'-S''-2S\phi'^2=\frac{k^2S}{4}\left(l'+\frac{1}{2\eta}\right)^2\\
\label{FE d}2S\phi'+S'=-|\eta|^{-1/2},\\
\label{FE e}\eta Sl'' +\eta S'l' + \frac{Sl'}{2}-\frac{S}{4\eta}=- \frac{2V_0}{k^2|\eta|}Se^{\lambda l},
\eea\end{subequations}
where $V_0=e^{c_1}\bar{V_0}$ is constant and $\lambda = k^2/2-1$. Equation \refb{FE e} is the wave equation for $\psi$ and is obtained from $\nabla^a\nabla_a\psi-V'(\psi)=0$. Region \textbf{I} of the spacetime corresponds to the interval $\eta\in[0,1]$, with the axis at $\eta=1$ and $\cal{N}_-$ at $\eta=0$. The regular axis conditions for the metric functions were found to be \cite{Part1}
\beq
\label{axis} S(1)=0,\quad S'(1)=-1,\quad \gamma'(1)=0,\quad \phi'(1)=-1/4.
\eeq
For values $(k^2,V_0,l_0)\in K$, solutions exist throughout region \textbf{I}, and $\cal{N}_-$ is a regular spacetime hypersurface. These solutions, which are the subject of this paper, may be extended into region \textbf{II}, which corresponds to $\eta\in(-\infty,0]$. It is assumed that $(k^2,V_0,l_0)\in K$ for the remainder of the paper. Notice that, in particular, we have $k^2<2$, or equivalently $\lambda<0$ and $\LL=1-k^2/2<1$. Note that $\cal{N_-}$ is at $\eta=0$ and $\cal{N}_+$ is at $u=0,v\in[0,\infty)$. Hence, $\eta\rightarrow-\infty$ everywhere on $\cal{N}_+$, approaching from inside region \textbf{II}. For the remainder of this paper, when we take the limit $\eta\rightarrow-\infty$, it is implied that we are taking the limit $u\rightarrow 0$ along lines of constant $v>0$. 
Our aim is to determine whether or not $\mathcal{N}_+$ exists as part of the spacetime, which answers the question of whether the singularity is naked or not. This coordinate layout  is illustrated in Figure 2.
We work with a rescaling of the similarity variable, which replaces (10) with an autonomous system, and adopt a dynamical systems approach. 

\begin{center}\begin{figure}[h]
\includegraphics[scale=0.5]{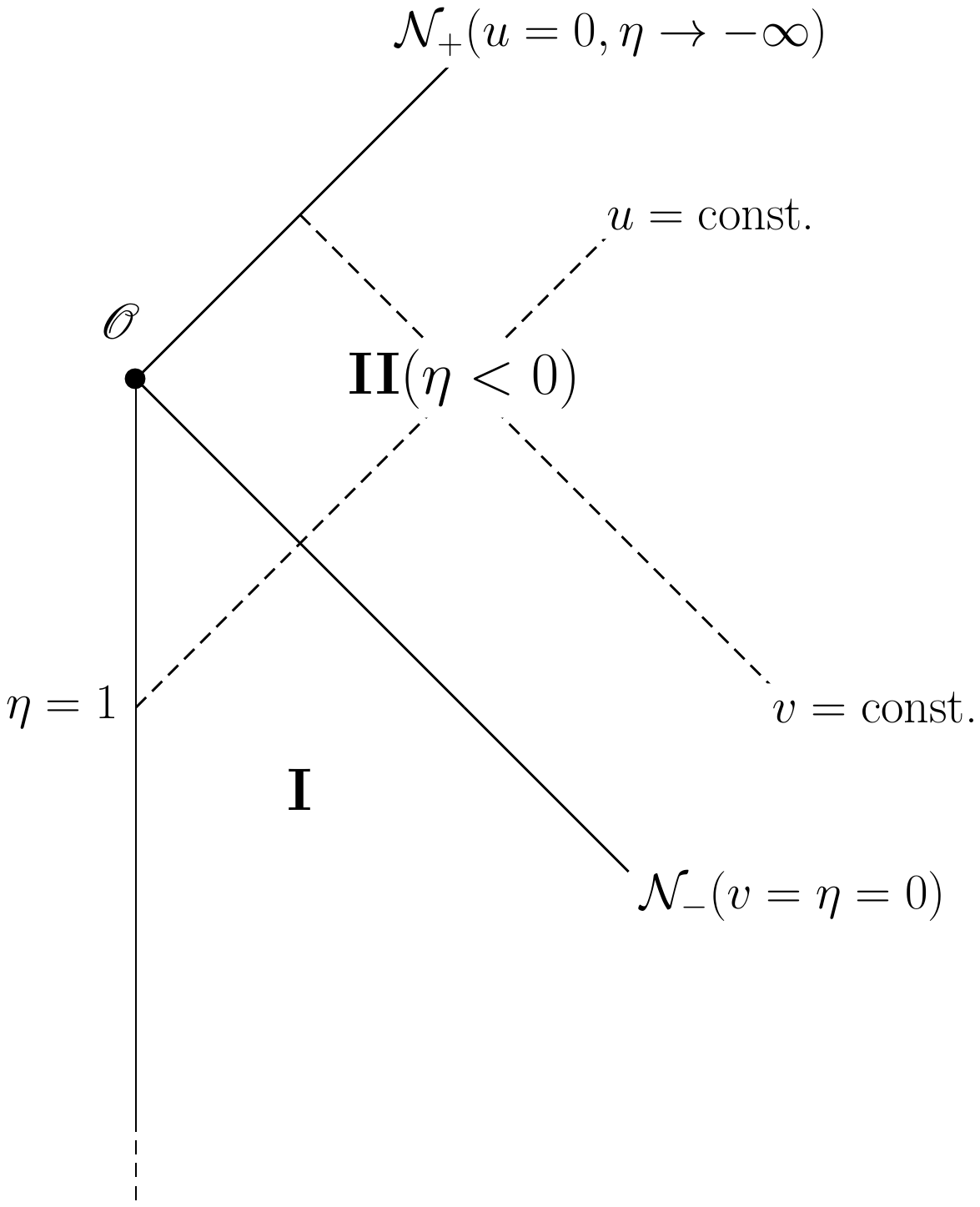}
\vskip -150pt
\caption{Coordinate layout. Our central question is whether or not ${\cal{N}}_+$ is part of the spacetime. }
\end{figure}\end{center}

\begin{proposition} Let $t=\log(-\eta),\delta = {\rm sgn}(V_0), \sigma(t)=S(\eta)$ and 
\bea
\label{x_i}x_0(t)=\frac{e^{t/2}}{\sigma(t)},\qquad x_1(t)=\frac{\eta S'(\eta)}{S}=\frac{\sigma'(t)}{\sigma(t)},\\ 
\nonumber x_2(t) = |V_0|e^{\lambda l(\eta)},\qquad x_3(t) =\eta l'(\eta)+\frac{1}{2}=\frac{dl}{dt}+\frac{1}{2}.
\eea
Then $x_0,x_1,x_2,x_3$ satisfy
\begin{subequations}
\begin{eqnarray}
\label{x1}x_1'(t)&=x_1+\delta x_2-x_1^2,\\
\label{x2}x_2'(t)&=\LL\left(\frac{1}{2}-x_3\right)x_2,\\
\label{x3}x_3'(t)&=\frac{x_3}{2}+\frac{x_1}{2}+\delta\frac{2x_2}{k^2}-x_1x_3,\\
\label{x0}x_1^2-&x_0^2-\left(\frac{k^2}{2}+1\right)x_1-\frac{k^2x_3^2}{2} +k^2x_1x_3-2\delta x_2=0,
\end{eqnarray}
\beq
\label{origin}\lmm(x_0,x_1,x_2,x_3) = (0,0,0,0).\eeq
\end{subequations}\end{proposition}
\begin{proof} First note that \refb{x2} comes directly from the definitions of $x_2$ and $x_3$. Given $f(\eta)$, defining $F(t)=f(\eta)$ yields $\eta f'(\eta)=F'(t)$ and $\eta^2f''(\eta)=F''(t)-F'(t)$. Equations \refb{x1} and \refb{x3} follow directly from \refb{FE b} and \refb{FE e}. Equation \refb{FE d} is equivalent to 
\beq
\label{phi'2}\frac{d\phi}{dt}=\frac{x_0-x_1}{2}.
\eeq
Differentiating \refb{FE a} with respect to $t$ gives 
\beq
\label{gam'2}2\frac{d\gamma}{dt}=-2\frac{d\phi}{dt}+\frac{k^2}{2}\frac{dl}{dt}-\frac{1}{2}= x_1-x_0+\frac{k^2x_3}{2}-\frac{k^2}{4}-\frac{1}{2}. 
\eeq
Dividing \refb{FE c} by $S$, changing variables and replacing $d\gamma/dt$ and $d\phi/dt$ using \refb{phi'2} and \refb{gam'2} produces
\beq
x_1\left(x_1-x_0+\frac{k^2x_3}{2}-\frac{k^2}{4}-\frac{1}{2}\right)-\delta x_2-\frac{1}{2}\left(x_0-x_1\right)^2=\frac{k^2x_3^2}{4}. 
\eeq
Multiplying by 2 and simplifying gives \refb{x0}. It was shown in \cite{Part1} that 
\beq
\label{lim 0}\lim_{\eta\rightarrow0+}\left(\eta l'(\eta),|V_0|e^{\lambda l}, \frac{\eta S'(\eta)}{S}\right) = \left(-\frac{1}{2},0,0\right),
\eeq
and that $S$ is non-zero and finite at $\eta=0$. The condition \refb{origin} follows immediately. 
\end{proof}
\noindent We note that the equations \refb{x1}-\refb{x3} subject to \refb{origin} define a dynamical system and may be studied independently of \refb{x0}. 
\section{Asymptotic behaviour of solutions at $\mathcal{N}_-$}
\begin{proposition}Let 
\begin{subequations}\bea
\label{mu}\mu_1= x_1+Ax_2,\qquad \mu_3= x_3+Bx_2,\\
\label{AB}A=\delta\frac{4}{2+k^2},\qquad B=\delta\frac{16}{k^4(2+k^2)}.
\eea
Then $\mu_1,\mu_3$ satisfy
\bea
\label{mu1'}\mu_1'=\mu_1-x_1^2-A\LL x_2x_3,\\
\label{mu3'}\mu_3'=\frac{\mu_3}{2}+\frac{\mu_1}{2}-x_1x_3-B\LL x_2x_3.
\eea
\end{subequations}\end{proposition}
\begin{proof}
It is straightforward to check that \refb{mu1'},\refb{mu3'} follow directly from \refb{mu},\refb{AB} and \refb{x1}-\refb{x3}. 
\end{proof}
\noindent We make use of the following result, which may be found in chapter 9 of \cite{Hartman}.
\begin{theorem} In the differential equation
\beq
\label{ODE}{\boldsymbol{x}}'(t)=F{\boldsymbol{x}}+G({\boldsymbol{x}}),
\eeq
let $G(\boldsymbol{x})$ be of class $C^1$ with $G(0)=0, \partial_{\boldsymbol{x}}G(0)=0$. Let the constant matrix $F$ possess $d>0$ eigenvalues having positive real parts, say, $d_i$ eigenvalues with real parts equal to $\alpha_i$, where $\alpha_1>\ldots>\alpha_r>0$ and $d_1+\ldots+d_r=d,$ whereas the other eigenvalues, if any, have non-positive real parts. If $0<\omega<\alpha_r,$ then \refb{ODE} has solutions ${\boldsymbol{x}=\boldsymbol{x}}(t)$$\neq0$, satisfying
\beq 
\label{H1}||{\boldsymbol{x}}(t)||e^{-\omega t} \to 0,\qquad \mbox{as}\qquad t\rightarrow-\infty,
\eeq
where $||{\boldsymbol{x}}(t)||$ denotes the Euclidean norm, and any such solution satisfies
\beq 
\label{H2}\lim_{t\rightarrow-\infty}t^{-1}\log||{\boldsymbol{x}}(t)||=\alpha_i,\qquad\mbox{for some }i.
\eeq\end{theorem}
\hbox{}\hfill$\square$\\
\noindent We define the vector $\boldsymbol{x}$ by 
\beq
\boldsymbol{x}=(x_1,x_2,x_3). 
\eeq
The system defined by \refb{x1}-\refb{x3} and \refb{origin} satisfies the hypothesis of this theorem, which grants local existence of solutions near the origin of the $\boldsymbol{x}$-system, which is at $t=-\infty$. We denote by $(-\infty,t_M)$ the maximal interval of existence for a given solution. 
\begin{lemma}For any $\eps>0$, there exists $T(\eps)\in(-\infty,t_M)$ such that
\beq
|x_i|< e^{(\LL/2-\eps)t},
\eeq 
for $t<T(\eps)$ and each $i\in\{1,2,3\}$.
\end{lemma}
\begin{proof} The system defined by \refb{x1}-\refb{x3} is of the form \refb{ODE}, where the matrix
\beq F=\left(\begin{array}{ccc}1&\delta&0\\
0&\LL/{2}&0\\
1/{2}&{2\delta}/{k^2}&{1}/{2}\end{array}\right)\eeq
has 3 positive eigenvalues, $\LL/2,1/2$ and $1$, of which $\LL/2$ is the smallest. Solutions to \refb{x1}-\refb{x3} therefore exist, which satisfy \refb{H1} and \refb{H2}. Using (\ref{H2}), for any $\eps>0$, there exists $T(\eps)<0$ such that 
\beq
\log||\boldsymbol{x}(t)||<(\LL/2-\eps)t,
\eeq
for all $t<T(\eps)$. Since $|x_i|\le||\boldsymbol{x}||$ for each $i\in\{1,2,3\}$, the result follows.
\end{proof}
\begin{lemma} For $\delta=1\,( \mbox{\emph{respectively}}\, \delta=-1)$, there exists $T\in(-\infty,t_M)$ such that $x_1<0,x_3<0\,(\mbox{\emph{respectively}}\, x_1>0,x_3>0)$ for $t\in(-\infty,T)$.
\end{lemma}
\begin{proof} Using Lemma 3.1 we have $|x_i|=O(e^{(\LL/2-\eps)t})$ in the limit $t\rightarrow-\infty$, for any $\eps>0$. From \refb{mu1'} we then have 
\beq
\frac{d}{dt}\left(e^{-t}\mu_1\right)= -e^{-t}(x_1^2+A\LL x_2x_3)=O(e^{(\LL-2\eps-1)t})\quad\mbox{as}\quad t\rightarrow-\infty,
\eeq
which may be integrated to give
\beq
\mu_1=c_2e^t+O(e^{(\LL-2\eps)t})=O(e^{(\LL-2\eps)t}),
\eeq
and so 
\beq
\label{x1 asym}x_1=-Ax_2+O(e^{(\LL-2\eps)t})\quad\mbox{as}\quad t\rightarrow-\infty,
\eeq
for some constant $c_2$, by choosing $\epsilon>0$ so that $\LL-2\eps<1$ (recall that $\LL<1$). A similar process using \refb{mu3'} yields
\beq
\mu_3= O(e^{mt}),
\eeq
and so 
\beq
\label{x3 asym}x_3 = -Bx_2 +O(e^{mt})\quad\mbox{as}\quad t\rightarrow-\infty,
\eeq
where $m=\min\{1/2, \LL-2\eps\}$. Since $\lmm x_3=0$, we may choose $T(\eps)$ such that $|\lambda x_3|<\eps$ for $t<T(\eps)$. We then have 
$$
\frac{x_2'}{x_2}<\frac{\LL}{2}+\eps, \qquad \mbox{for } t\in(-\infty,T(\eps)).
$$
Integrating over $[t,T]$ shows that $x_2(t)>x_2(T)e^{(\LL/2+\eps)(t-T)}$ on the same interval. Choosing $\eps$ such that $\LL/2+\eps<\min\{1/2,\LL-2\eps\}$ shows that the $x_2$ terms in equations \refb{x1 asym} and \refb{x3 asym} are dominant for $t$ sufficiently close to $-\infty$. $T$ may be then chosen, without loss of generality, such that $x_1$ and $x_3$ have the same sign as $-Ax_2$ and $-Bx_2$ on $(-\infty,T)$, respectively. Note from \refb{AB} that $A$ and $B$ have the same sign as $\delta$.
\end{proof}
\begin{proposition} There exists $c_3>0$ such that 
\beq
\lim_{t\rightarrow-\infty} e^{-\LL t/2}{\boldsymbol{x}}= c_3\left(A,1,B\right).
\eeq
\end{proposition}
\begin{proof} Integrating \refb{x2} over $[t,T]$ we have 
\beq
e^{-\LL t/2}x_2(t)=e^{-\LL T/2}x_2(T)+\int_{t}^{T}e^{-\LL t'/2}\LL x_2x_3\,dt'.
\eeq
Consider the case $\delta=-1$. By Lemma 3.2 we have $x_3>0$ on $t\in(-\infty,T)$, and by choosing $T$ sufficiently small such that the bounds of Lemma 3.1 hold, we have
\beq
e^{-\LL T/2}x_2(T)<e^{-\LL t/2}x_2(t)<e^{-\LL T/2}x_2(T)+\int_t^{T}\LL e^{(\LL /2-2\eps)t'}\,dt'.
\eeq
The integral here is finite in the limit $t\rightarrow-\infty$ for $\eps<\LL/4$ and so $e^{-\LL t/2}x_2$ has positive and finite upper and lower bounds in the limit as $t\rightarrow-\infty$. It is also monotone for $t<T$ and so we have $\lmm e^{-\LL t/2}x_2=c_3>0$, for some $c_3>0$. A similar argument gives this result in the case $\delta=1$. Multiplying \refb{x1 asym} and \refb{x3 asym} by $e^{-\LL t/2}$ and taking the limit $t\rightarrow-\infty$ gives $\lmm e^{-\LL t/2}\boldsymbol{x}=c_3(A,1,B)$. 
\end{proof}
\noindent\textbf{Comment 3.1.} For convenience, we define $t_*$ by $c_3=e^{-\LL t_*/2}$. Notice then that the result of 3.2 may be written as 
$\lim_{\bar{t}\rightarrow-\infty} e^{-\LL \bar{t}/2}{\boldsymbol{x}}= \left(A,1,B\right)$ where $\bar{t}=t-t_*$. Noting that \refb{x1}-\refb{x3} is invariant under translations of the independent variable we drop the bar and let $\bar{t}=t$. Hence
\beq
\lmm e^{-\LL t/2}{\boldsymbol{x}}= \left(A,1,B\right).
\eeq
This describes the asymptotic behaviour of solutions to the future of $\cal{N}_-$, as they emerge from $\cal{N}_-$. 


\section{Analysis of fixed points}

\begin{proposition}
The equilibrium points of the system (\ref{x1})-(\ref{x3}) are given by 
\beq
P_1=(1,0,1),\quad P_2=\left(\alpha_-,\frac{-\delta k^2}{8},\frac{1}{2}\right),\quad P_3= \left(\alpha_+,\frac{-\delta k^2}{8},\frac{1}{2}\right),
\eeq
where 
\beq
\alpha_\pm=\frac{1\pm\sqrt{\LL}}{2}.
\eeq
\end{proposition}

\begin{proof}
This is straightforward to check. 
\end{proof}

\begin{proposition}
Define $f(\boldsymbol{x})$ by setting ${\boldsymbol{x}}'(t)=f({\boldsymbol{x}})$, where the component equations are given by \refb{x1}-\refb{x3}. Let 
\beq
{\boldsymbol{y}}(s)=(y_1,y_2,y_3), \qquad \begin{array}{ll} y_1(s)=1-x_1(t)\\y_2(s)=x_2(t)\\y_3(s)=1-x_3(t) \end{array} \qquad s=-t.
\eeq
Then 
\beq
{\boldsymbol{y}}'(s)=f({\boldsymbol{y}}). 
\eeq
\end{proposition}
\begin{proof}
This is straightforward to check. 
\end{proof}
\begin{proposition}
Suppose $\lmp {\boldsymbol{x}}=P_1$. Then 
\beq
\label{lim 1-x}\lmp e^{\LL t/2}(1-x_1,x_2,1-x_3) = c(A,1,B),
\eeq
for some constant $c>0$. 
\end{proposition}
\begin{proof}
First note that $\lmp \boldsymbol{x}=P_1$ is equivalent to $\lim_{s\rightarrow-\infty}\boldsymbol{y}= (0,0,0)$. Since $\boldsymbol{y}'(s)=F(\boldsymbol{y})$, solutions emanating from the origin of the $y$-system satisfy the exact conditions satisfied by solutions emanating from the origin of the $x$-system used in the proofs of Section 3. We may, therefore, carry out an identical analysis to find
\beq
\lim_{s\rightarrow-\infty}e^{-\LL s/2}\boldsymbol{y}=c(A,1,B), 
\eeq
which is our result. 
\end{proof}
\begin{proposition}
Let $k^2<2$ with $k^2\neq \sqrt{3}-1$. Suppose that $\lmp {\boldsymbol{x}}=P_1$. Then $\lmp r$ is non-zero and finite and $\mathcal{R}$ is bounded in this limit, where $\cal{R}$ is the Ricci scalar corresponding to the line element \refb{LE eta} and $r=|u|\sigma$ is the radius of the cylinders in this spacetime. 
\end{proposition}
\begin{proof}
Using Proposition 4.3, for any $\epsilon>0$ there exists $T(\epsilon)$ such that 
\beq
\label{x1 bound}1-(Ac+\epsilon)e^{-\LL t/2}<x_1<1-(Ac-\epsilon)e^{-\LL t/2},
\eeq
for $t>T(\epsilon)$. Recalling $x_1=\sigma'/\sigma$, it is straightforward to show that this leads to 
\beq
\label{eS}C_1<e^{-t}\sigma<C_2,\qquad t>T(\epsilon),
\eeq
for positive constants $C_1,C_2$. We also have
\beq
\frac{d}{dt}(e^{-t}\sigma)= e^{-t}\sigma\left(x_1-1\right). 
\eeq
Using \refb{x1 bound} shows that $e^{-t}\sigma$ is monotone decreasing near $P_1$. Moreover, this may be integrated using \refb{x1 bound} to show that $e^{-t}\sigma$ has a finite, non-zero limit as $t\rightarrow\infty$. In region \textbf{II} of the spacetime we have $v>0,u<0$ and thus $|u|=-v/\eta= ve^{-t}$. It follows that $r=|u|\sigma=ve^{-t}\sigma$ has a positive finite limit approaching ${\cal{N}}_+\,(u=0)$ along lines of constant $v$. It follows from \refb{phi'2},\refb{x1 bound} and \refb{eS} that
\beq
\frac{1}{C_2e^{t/2}}+(Ac-\epsilon)e^{-\LL t/2} <2\frac{d\phi}{dt}+1 < \frac{1}{C_1e^{t/2}} +(Ac+\epsilon)e^{-\LL t/2},
\eeq
for $t>T(\eps)$. We see that if $\eps\le Ac$ then $2\phi+t$ is monotone in $t$. Integrating and taking exponentials then shows that $\lmp e^{2\phi+t}$ exists, is non-zero and finite. Hence, $\lmp|u|e^{-2\phi}=\lmp ve^{-2\phi-t}$ is non-zero and finite. So far we have shown that $g_{\theta\theta}=|u|^{-1}e^{2\phi}r^2$ and $g_{zz}=|u|e^{-2\phi}$ have non-zero, finite limits as $t\rightarrow+\infty$. Using similar arguments, it may shown that the metric component $|u|^{-1}e^{2\gamma+2\phi}$ behaves like $e^{(1-\LL/2)t}$ in the limit as $t\rightarrow +\infty$ and, therefore, has limit $+\infty$. However, by making the coordinate transformation $\bar{u}=-2|u|^{\LL /2}/\LL$ we avoid this problem. The corresponding metric component in this coordinate system is $|\bar{u}|^{-1}e^{2\gamma+2\phi}$ and it may be shown in a similar fashion that this has a non-zero, finite limit as $t\rightarrow+\infty$.
In \cite{Part1} it was shown that the Ricci scalar may be written as 
\beq
\label{Ricci}\mathcal{R}=\frac{e^{-k^2l/2+t/2-c_1}}{v}\left(\frac{k^2}{2}\left(1-x_3\right)x_3-4\delta x_2\right). 
\eeq
It may be shown, using \refb{lim 1-x} in a similar way, that for all sufficiently large $t$, we have 
\beq
\label{e^l bound}C_3e^{\LL t/2} <e^{-k^2l/2+t/2-c_1}<C_4e^{\LL t/2},
\eeq
for some positive constants $C_3,C_4$. (To obtain this result, we integrate the third component of the vector in \refb{lim 1-x} at large $t$ to obtain
\beq e^{c_1}\exp[O(e^{\lambda t/2})]<e^{-(l-t/2)}<e^{c_2}\exp[O(e^{\lambda t/2})] \eeq
and combine with the second component of \refb{lim 1-x}.) We also have 
\bea
\lmp e^{\LL t/2}\left(\frac{k^2}{2}\left(1-x_3\right)x_3-4\delta x_2\right)&= \frac{k^2Bc}{2}-4\delta c\\
\nonumber &= \left(\frac{8}{k^2(2+k^2)}-4\right)\delta c \neq 0, 
\eea
for $k^2\neq\sqrt{3}-1$, using \refb{lim 1-x}. Combining this with \refb{Ricci} and \refb{e^l bound} shows that $\lmp \mathcal{R}$ is bounded for essentially all $k^2<2$.
\end{proof}
This result shows that in spacetimes where the solutions to the field equations satisfy $\lmp\boldsymbol{x}=P_1$, the future null cone of the singularity $\cal{N}_+$ is regular and exists are part of the spacetime, thus rendering the singularity at the origin naked. However, it is shown in later sections that none of the solutions actually do evolve to $P_1$. 
\begin{proposition} 
If $\lmp {\boldsymbol{x}}=P_2$ or $\lmp{\boldsymbol{x}}=P_3$, then $\lmp r=0$ and $\lmp\mathcal{R}=+\infty$, where $r$ is the radius of the cylinders and $\mathcal{R}$ is the Ricci scalar. 
\end{proposition}
\begin{proof}
If $\lmp x_1=\alpha_\pm$ then for any $\eps>0$ there exists $T(\eps)$ such that $x_1<\alpha_++\eps$ for $t>T(\eps)$, since $\alpha_-<\alpha_+$. Note that $\sqrt{|\lambda|}=\sqrt{1-k^2/2} < 1-k^2/4$, which gives $\alpha_+<1-k^2/8$. This leads to $\sigma<\sigma(T)e^{(1-k^2/8+\eps)(t-T)}$ for $t>T$.  It follows that $r=|u|\sigma<v\sigma(T)e^{(-k^2/8+\eps)(t-T)}$ for $t>T$. Choosing $\eps<k^2/8$ shows that $\lmp r=0$, for $v\in(0,\infty)$. It is straightforward to show that $\lmp e^{-k^2l/2+t/2}=+\infty$ follows from $\lmp dl/dt=0$, which is equivalent to $\lmp x_3=1/2$. Then using $\lmp x_2=k^2/8$ and $\delta=-1$ we find that $\lmp \mathcal{R}=+\infty$.
\end{proof}
\begin{proposition}
Let $\delta=-1$. Then there is no solution of \refb{x1}-\refb{x0} which satisfies $\lmp \boldsymbol{x}=P_2$. 
\end{proposition}
\begin{proof}
$x_0$ satisfies
\beq
x_0'=x_0\left(\frac{1}{2}-x_1\right). 
\eeq
If $\lmp x_1 = 1/2-\sqrt{\LL}/2$ then $x_0'\sim (\sqrt{\LL}/2)x_0$ as $t\rightarrow+\infty$. Since $x_0>0$ for $t>-\infty$ we must have $\lmp x_0=+\infty$. This contradicts $\lim\boldsymbol{x}=P_2$ and \refb{x0}.  
\end{proof}
\begin{proposition}
Let $\delta=+1$. Then there is no solution of \refb{x1}-\refb{x0} which satisfies $\lmp \boldsymbol{x}=P_2$ or $\lmp\boldsymbol{x}=P_3$.
\end{proposition}
\begin{proof}
This follows immediately from the fact that $x_2>0$. \end{proof}

\noindent\textbf{Comment 4.1} We note that $\lmt\boldsymbol{x}=P_1$ or $P_3$ are consistent with \refb{x0}. 

\section{Global behaviour of solutions of the dynamical system}

Our aim in this section is to give a complete account of the future evolution of solutions of the dynamical system (\ref{x1})-(\ref{origin}) in the case $\lambda<0$ (corresponding to Case 3 in Figure 1). Our conclusion, given in Propositions 5.1 and 5.2 below, is that in every case, the maximal interval of existence is bounded above: the solutions only exist for a finite time in the future. We note that as the solutions evolve from $t=-\infty$, the maximal interval of existence must have the form $(-\infty,t_M)$ for some $t_M\le+\infty$. The key conclusion that we make is that $t_M$ is finite in every case. 

The argument is structured as follows. The first important result is Lemma 5.2, where we deduce that the state variables $x_i$, $i=1,2,3$ are monotone in a neighbourhood of $t=t_M$. This requires Lemma 5.1. In Lemmas 5.3 - 5.6, we establish connections between the limits of various state variables as $t\to t_M$. Lemmas 5.7 - 5.11 are linked by the theme of finding precursors to $t_M$ being finite. Among these is the important Lemma 5.8 which provides restrictions on possible limits of some of the key state variables as $t\to t_M$. Propositions 5.1 and 5.2 then establish our main result, that $t_M$ is indeed finite.

To begin, we quote the following standard result which is helpful in determining the maximal intervals of existence (see, for example, \cite{Tav}, ch.4).

\begin{theorem} 
Let $\Psi_{\bsy{a}}(t)$ be the unique solution of the differential equation ${\boldsymbol{x}}'={\boldsymbol{f}(\boldsymbol{x})}$, where ${\bsy{f}}\in C^1(\mathbb{R}^n)$, which satisfies ${\boldsymbol{x}}(0)={\bsy{a}}$, and let $(t_{\min},t_{\max})$ be the maximal interval of existence on which $\Psi_{\bsy{a}}(t)$ is defined. If $t_{\max}$ is finite, then 
\beq
\lim_{t\rightarrow t^-_{\max}}||\Psi_{\bsy{a}}(t)||=+\infty. 
\eeq
\end{theorem}

This theorem tells us that solutions exist while each component of the solution is finite. Recall that our maximal interval of existence has the form $(t_{\min},t_{\max})=(-\infty, t_M)$.

\begin{lemma}
For $V_0<0$ (so that $\delta=-1$), suppose there exists $t_0\in(-\infty,t_M)$ such that $x_3(t_0)=1/2$ and $x_3<1/2$ for $t\in(-\infty,t_0)$. Then $x_2(t_0)< k^2/8$ and $x_3(t)>1/2$, $x_2(t)< k^2/8$ hold for all $t\in(t_0,t_M)$. 
\end{lemma}
\begin{proof}
First note that $x_3=1/2,x_2=k^2/8$ defines an invariant manifold of the system \refb{x1}-\refb{x3}, so if $x_3(t_0)=1/2$, $x_2(t_0)=k^2/8$, then we would have $x_3=1/2,x_2=k^2/8$ for all $t\in(-\infty,t_M)$, which is clearly not the case, since $x_2,x_3\to0$ as $t\to-\infty$. Moreover, at $x_3=1/2$ we have
\beq
\label{x3=1/2} x_3'=\frac{1}{4}-\frac{2x_2}{k^2},
\eeq
and so $x_3$ cannot reach $1/2$ from below if $x_2(t_0)>k^2/8$. Hence, we must have $x_2(t_0)<k^2/8$ and $x_3'(t_0)>0$. Equation \refb{x3=1/2} also shows that $x_3$ cannot cross $1/2$ from above if $x_2<k^2/8$. Given that $x_2$ is decreasing if $x_3>1/2$, we must have $x_3>1/2$ and $x_2<k^2/8$ for all $t\in(t_0,t_M)$. 
\end{proof}

This leads us to an important monotonicity result:

\begin{lemma} Let $i\in\{1,2,3\}$. Then each $x_i$ is monotone in the limit as $t\rightarrow t_M^-$. Hence, either $\lmt x_i$ exists or $\lmt x_i=\pm\infty$.
\end{lemma}
\begin{proof}
Lemma 5.1 tells us that if $\delta=-1$, then $x_3-1/2$ can only change sign once. If $\delta=1$, then at $x_3=1/2$ we have $x_3'=1/4+2x_2/k^2>0$, so $x_3-1/2$, and thus $x_2'$, can only change sign once in this case also. At $x_1'=0$ we have $x_1'' = \delta x_2'$, which means that $x_1'$ can only change sign twice. At $x_1-x_3=0$ we have $x_1'-x_3'= \delta(1-2/k^2)x_2$ which always has the same sign, specifically, the opposite sign to $\delta$. Hence, $x_1-x_3$ can only change sign once also. Now, at $x_3'=0$ we have
\beq
x_3''=\left(\frac{1}{2}-x_3\right)\left(\frac{1}{2}-x_1\right)\left(x_1-x_3\right).
\eeq
The right hand side here may only change sign a finite number of times. Hence, $x_3'$ eventually becomes fixed in sign and $x_3$ becomes monotone. 
\end{proof}

\begin{lemma} Let $t_0\in(-\infty,t_M)$ satisfy $\sigma(t_0)\neq 0$. Then 
\beq
\label{S=0}\lmt \sigma = 0 \quad \Leftrightarrow \quad \lmt\int_{t_0}^t x_1\,dt' = -\infty.
\eeq
Furthermore, if $t_M<\infty$ then 
\beq
\label{eq}\lmt\sigma=0 \Rightarrow \lmt x_1=-\infty.
\eeq
\end{lemma}
\begin{proof}
By the definition of $x_1$ we have
\beq
\sigma = \sigma(t_0)\exp\left(\int_{t_0}^tx_1\,dt'\right), 
\eeq
from which \refb{S=0} immediately follows. To establish \refb{eq} we note that since $t_M<\infty$, divergence of the integral $\int_{t_0}^tx_1\,dt'$ as $t\to t_M^-$ implies the divergence of the integrand in this limit. 
\end{proof}

\begin{lemma} If $t_M$ is finite and $\lmt x_1 = -\infty$ then $\lmt \sigma=0$ .
\end{lemma}
\begin{proof}
If $\lmt x_1=-\infty$ then either $\lmt \sigma'=-\infty$ or $\lmt \sigma=0$ (see \refb{x_i}). Writing \refb{x1} in terms of $x_2$ and $\sigma$ gives
\beq
\label{sig''}\sigma''=\sigma'+\delta x_2 \sigma.
\eeq
If $\delta=1$ then we have $\sigma''>\sigma'$, which rules out $\lmt\sigma'=-\infty$, since $t_M$ is finite. 
In the case $\delta=-1$, suppose that $x_2$ is bounded above by a constant $b$ for all $t\in(-\infty,t_M]$ and let $\sigma_M$ be the maximum of $\sigma$ on this interval. 
Then we have $\sigma''>\sigma'-b\sigma_M$, which also rules out $\lmt\sigma'=-\infty$. 

We now turn to the case $\delta=-1,\lmt x_2=+\infty$ (monotonicty, established in Lemma 5.2, leaves this as the only remaining option). Consider 
\beq
x_1'-\frac{k^2x_3'}{2}= \frac{\LL x_1}{2}+\left(\frac{1}{2}-x_1\right)\left(x_1-\frac{k^2x_3}{2}\right). 
\eeq
Using $e^{-t/2}\sigma$ as integrating factor we find that
\beq
\label{bound}\frac{d}{dt}\left(e^{-t/2}\sigma\left(x_1-\frac{k^2x_3}{2}\right)\right) = \frac{\LL }{2}e^{-t/2}\sigma'<0,
\eeq
where the inequality holds on some interval $(t_0,t_M)$.
Assuming $\lmt \sigma>0$, integrating then shows that $x_1-k^2x_3/2$ is bounded above for all $t\in[t_0,t_M]$. We then have
\beq
\frac{x_2'}{x_2}<\LL\left(\frac{1}{2}+\frac{2b}{k^2}-\frac{2x_1}{k^2}\right),
\eeq
for some constant $b$. 
Integrating shows that since $\lmt x_2=+\infty$ then $\lmt\int_{t_0}^t x_1=-\infty$, which gives $\lmt \sigma=0$ by Lemma 5.3. 
\end{proof}

We note that Lemmas 5.3 and 5.4 tell us that, for $t_M<\infty$, $\lmt \sigma = 0 $ if and only if $\lmt x_1=-\infty$. 

\begin{lemma} If $t_M$ is finite and $\lmt x_1=-\infty$, then $\lmt x_3=\pm\infty$ or $\lmt x_3=1/2$.
\end{lemma}
\begin{proof}
Integrating \refb{x3} we have 
\beq
\label{int x3}x_3= e^{(t-t_0)/2}x_3(t_0)+\int_{t_0}^t e^{(t-t')/2}\left(\left(\frac{1}{2}-x_3\right)x_1 +\delta\frac{2x_2}{k^2}\right)\,dt',
\eeq
for any $t_0\in(-\infty,t_M)$. Given that $\lmt x_1=-\infty$, we have
\beq
\label{int x1} \lmt\int_{t_0}^t x_1\,dt' = -\infty,
\eeq
by Lemmas 5.3 and 5.4. Assuming $|x_3|$ is bounded, which gives $x_2$ bounded by (\ref{x2}), we must have $\lmt x_3=1/2$ by inspection of \refb{int x3} with \refb{int x1}. If $x_3$ is unbounded then we must have $\lmt x_3=\pm\infty$ by Lemma 5.2. 
\end{proof}

\begin{lemma} If $t_M$ is finite and $\lmt x_3 = \pm\infty$, then $\lmt x_1=-\infty$.
\end{lemma}
\begin{proof}
Suppose that $\lmt x_3 = +\infty$. Then $\lmt x_2<\infty$, by \refb{x2}, and it follows from \refb{x3} that $x_3'<(\frac12-x_1)x_3+\frac{x_1}{2}+b$ for all $t\in(-\infty,t_M)$ and some constant $b$.  This yields the inequality
\begin{equation}
x_3(t) < \frac{\mu(t)}{\mu(t_0)}x_3(t_0)+\int_{t_0}^t \frac{\mu(t')}{\mu(t)}(\frac{x_1}{2}+b)dt',\end{equation} where 
\begin{equation}
\mu(t) = \exp\left\{\int_{t_0}^t(x_1-\frac12)\right\}.\end{equation}
Since $t_M<\infty$ we must have $\lmt x_1=-\infty$ in order that $x_3\to+\infty$. \\
Now suppose $\lmt x_3=-\infty$ and $\lmt x_1>-\infty$. Then we have $\lmt \sigma>0$ and it can be easily shown via \refb{bound} that $x_1-k^2x_3/2$ is bounded above. 
It follows that if $\lmt x_3=-\infty$ then $\lmt x_1=-\infty$.
\end{proof}

\begin{lemma} If $V_0<0$ (respectively $V_0>0$), then $x_1(t)<1$ (respectively $x_1(t)<0$) for all $t\in(-\infty,t_M)$.
\end{lemma}
\begin{proof}
If $\delta=-1$ then it follows directly from \refb{x1} that $x_1$ cannot cross 1 from below. If $\delta=1$ then by Lemma 3.2 we have $x_1<0$ on an initial interval, say $(-\infty,t_0)$. Now suppose that $x_1(t_0)=0$. It is clear that $x_0'>0$ on $(-\infty,t_0)$ and so $x_0(t_0)>0$. At $t_0$, \refb{x0} with $\delta=1$ reduces to 
\beq
-x_0^2(t_0)-\frac{k^2x_3^2(t_0)}{2}-2x_2(t_0)=0,
\eeq
which clearly contradicts $x_0(t_0)>0$. Hence, no such $t_0$ exists. 
\end{proof}

\begin{lemma} If $t_M$ is finite then 
\beq
\label{lim}\lmt x_1=-\infty, \qquad \lmt \sigma=0,
\eeq
and either 
\beq
\lmt x_3 =\pm \infty, \quad \mbox{or} \quad \lmt x_3 = \frac{1}{2}.
\eeq
\end{lemma}
\begin{proof}
Using Theorem 5.1 and Lemma 5.2 we must have $\lmt |x_i|=+\infty$ for some $i\in\{1,2,3\}$. By (\ref{x2}), if $x_3$ is bounded and $t_M$ is finite, then $x_2$ is bounded. By Lemma 5.7, $x_1<1$ for all $t\in(-\infty,t_M)$, and so $x_1\to-\infty$ in this case. Alternatively, we must have $\lim_{t\to t_M^-}x_3=\pm\infty$ in which case $x_1\to-\infty$ by Lemma 5.6. Lemmas 5.4 and 5.5 complete the proof.
\end{proof}

\begin{lemma} For $V_0<0$, suppose there exists $t_0\in(-\infty,t_M)$ such that $x_1(t_0)<0$. Then $t_M$ is finite.
\end{lemma}
\begin{proof} If $x_1(t_0)<0$ and $\delta=-1$ then \refb{x1} yields $x_1'(t_0)<-x_1(t_0)^2$ and so $x_1<0$ persists. That is, $x_1<0$ and $x_1'<-x_1^2$ for all $t\in(t_0,t_M)$. Integrating shows that $x_1$ diverges to $-\infty$ in finite $t$. 
\end{proof}

\begin{lemma}
For $V_0<0$, suppose there exists $t_0\in(-\infty,t_M)$ such that $x_2(t_0)=k^2/8$ and $x_2<k^2/8$ for all $t\in(-\infty,t_0)$. Then $x_3(t_0)<1/2$ and $x_3<1/2$, $x_2>k^2/8$ hold for all $t\in(t_0,t_M)$. 
\end{lemma}
\begin{proof}
Using Lemma 5.1, we must have $x_3<1/2$ for all $t\in(-\infty,t_0]$. We also have $x_2'>0$ while $x_3<1/2$ and since $x_3$ cannot cross $1/2$ from below while $x_2> k^2/8$ then we must have $x_2>k^2/8$ and $x_3<1/2$ for all $t\in(t_0,t_M)$. 
\end{proof}

\begin{lemma}
For $V_0<0$, suppose there exists $t_0\in(-\infty,t_M)$ such that $x_1(t_0)=1/2$. Then $t_M$ is finite.
\end{lemma}
\begin{proof}
At $x_1=1/2$, equation \refb{x0} with $\delta=-1$ simplifies to 
\beq
\frac{1}{4}+x_0^2+\frac{k^2}{4}+\frac{k^2}{2}\left(x_3^2-x_3\right)=2x_2. 
\eeq
Using the fact the $x_3^2-x_3\ge-1/4$ we then have 
\beq
\label{x2(t0)} x_2(t_0)>\frac{1}{8}+\frac{k^2}{16}+\frac{x_0^2}{2}>\frac{k^2}{8}. 
\eeq
There must then exist $t_*\in(-\infty,t_0)$ such that $x_2(t_*)=k^2/8$ and $x_2(t)<k^2/8$ for all $t<t_*$. Using Lemma 5.10 we have $x_3<1/2$, and thus $x_2'>0$, for all $t\in(t_*,t_M)$. Using \refb{x2(t0)} we have
$x_2>1/8+k^2/16$, from which it follows that
\beq
\label{x3'}x_3'<\frac{\lambda}{4k^2}-\left(\frac{1}{2}-x_1\right)\left(\frac{1}{2}-x_3\right)<\frac{\lambda}{4k^2}+\frac{1}{2}\left(\frac{1}{2}-x_3\right)
\eeq
for all $t\in(t_0,t_M)$, where we have used $x_1<1$ (see Lemma 5.7). This shows that $x_3'<0$ if $x_3>1/2+\lambda/2k^2$. It follows that $x_3-1/2<m=\max\{x_3(t_0)-1/2,\lambda/2k^2\}<0$, which gives $x_2'>\lambda m x_2$ , for all $t\in(t_0,t_M)$. If $t_M=+\infty$ then $\lmt x_2=+\infty$ which would cause $x_1$ to become negative in finite $t$, contradicting Lemma 5.9.
\end{proof}

\begin{proposition}
If $V_0<0$, then $t_M$ is finite and
\beq
\lmt x_1=-\infty, \qquad \lmt \sigma=0,
\eeq
and either 
\beq
\lmt x_3 =\pm \infty, \quad \mbox{or} \quad \lmt x_3 = \frac{1}{2}.
\eeq\end{proposition}
\begin{proof}
The preceding lemma rules out the possibility that $\boldsymbol{x}$ limits to $P_1$ or $P_3$ as $t\rightarrow\infty$, 
since the $x_1$ components of $P_1$ and $P_3$ are greater than one half. Proposition 4.6 rules out the possibility that $\boldsymbol{x}$ limits to $P_2$. Taking note of Lemma 5.2 which rules out limit cycles and other behaviours, we see that we must either have $\lmp ||\boldsymbol{x}||=+\infty$ or $t_M$ finite with $\lmt ||\boldsymbol{x}||=\infty$. We may rule out the former case as follows. We can't have $\lmp x_1=-\infty$, because in that case there would exist $t_0<+\infty$ such that $x_1(t_0)<0$ and thus $t_M$ would be finite by Lemma 5.9. Nor can we have $\lmp x_2=+\infty$ since this would cause $x_1$ to become negative in finite $t$, via \refb{x1}, so we would have $t_M$ finite here also. This also rules out $\lmp x_3=-\infty$ since this would give $\lmp x_2=+\infty$, by \refb{x2}. Given that $x_2>0$ and $x_1<1$ by Lemma 5.7, this leaves the possibility that $\lmp x_3=+\infty$. However, it is easy to see that if $\lmp x_3=+\infty$ and $\lmp|x_1|<\infty$, then \refb{x0} is not satisfied, since in that case we have $\lmp x_2<\infty$ and the left hand side has limit $-\infty$. We must, therefore, have $t_M$ finite. Then Lemma 5.8 applies to give the limits stated.
\end{proof}

\begin{proposition}
If $V_0>0$, then $t_M$ is finite and
\beq
\lmt x_1=-\infty, \qquad \lmt \sigma=0,
\eeq
and either 
\beq
\lmt x_3 =\pm \infty, \quad \mbox{or} \quad \lmt x_3 = \frac{1}{2}.
\eeq
\end{proposition}
\begin{proof}
It is easily checked that \refb{x0} with $\delta=1$ may be written as
\beq
\left(1+\frac{k^2}{2}\right)\left(x_1-x_1^2\right)= -2x_2-x_0^2-\frac{k^2}{2}\left(x_1-x_3\right)^2,
\eeq
from which it follows that 
\bea
-\left(x_1-\frac{1}{2}\right)^2&<-\left(\frac{x_0}{\kappa}\right)^2,
\eea
and so
\bea
x_1-\frac{1}{2}&<-\frac{x_0}{\kappa},
\eea
where $\kappa^2=1+k^2/2$ and we have used $x_0>0$ and $x_1<0$ (which is given by Lemma 5.7). This is equivalent to 
\beq
x_0'>\frac{x_0^2}{\kappa}.
\eeq
Integrating shows that $x_0=e^{t/2}/\sigma$ blows up in finite time, and so $t_M$ is finite with $\lmt\sigma=0$. As in Proposition 5.1, the limits follow by Lemma 5.8.
\end{proof}


\section{Global structure and strong cosmic censorship}
The aim of this section is to prove a strong cosmic censorship theorem for the class of spacetimes considered here. Strong cosmic censorship is a statement about solutions of the Cauchy initial value problem in General Relativity (see p.305 of \cite{Wald}). Clearly, we are not dealing with the Cauchy problem here, but the spirit of the result is the same as that of strong cosmic censorship. We prove that the solutions considered here (which evolve from a regular axis rather than an initial data surface) are globally hyperbolic and $C^1$-inextendible. This follows from the results established below regarding radial null geodesics, and from an argument based on the behaviour of a certain invariant $E$ of the spacetime which depends only on the metric and its first derivatives. This invariant satisfies $\lim_{t\to t_M^-} E = +\infty$, proving $C^1$-inextendibility. (The use here of the quantity $E$ mirrors the use of the Hawking mass to prove the $C^1$-inextendibility of solutions of the Einstein-Maxwell-Scalar Field equations \cite{Dafermos}.)\\
The quantity $E$ in question is the $C-$energy defined by Thorne \cite{Thorne}. In a cylindrically symmetric spacetime $(M,g)$ with axial and translational Killing vectors $\bsy{\xi}_{(\theta)}$ and $\bsy{\xi}_{(z)}$, the circumferential radius $\rho$, the specific length $L$ and the areal radius $r$ are defined as
\beq
\label{specL}\rho^2 = \xi_{(\theta)a}\xi_{(\theta)}^a,\quad L^2= \xi_{(z)a}\xi_{(z)}^a,\quad r = \rho L.
\eeq
The $C$-energy is then defined as 
\beq
\label{C-en} E= \frac{1}{8}\left(1-L^{-2}g^{ab}\partial_ar\partial_br\right).
\eeq
As noted in \cite{HNN}, this does not yield a uniquely defined quantity in a given cylindrical spacetime. Furthermore, $E$ can blow up even in flat spacetime. However, as we will see below, this pathology is linked to the over-abundance of Killing vector fields (KVF's) in flat spacetime.\par
As we see from the definition, $E$ is not a function of the metric alone, but depends also on the KVF's:
\beq
E = E(g_{ab}, \bsy{\xi}_{(\theta)}, \bsy{\xi}_{(z)}).
\eeq
It should be more correctly understood as a function of the axis of a cylindrical spacetime, relative to a given translation along the axis. In our class of spacetimes, the axis - and corresponding KVF \cite{Carot} -  is given. Likewise, the definition of the class considered means that we have another KVF ($\bsy{\xi}_{(z)}$) that both commutes with and is orthogonal to $\bsy{\xi}_{(\theta)}$ (it is this orthogonality requirement that puts us in the class of \textit{whole} cylinder symmetry). However, in a given spacetime with whole cylinder symmetry, with the axis and axial KVF specified, the translational KVF is not necessarily uniquely defined. Consequently, the $C-$energy relative to the axis is not necessarily well-defined. See \cite{HNN} for a counter-example, which arises in flat spacetime. This presents a difficulty if we wish to make invariant statements about the spacetime in terms of the $C-$energy $E$. However, Proposition 6.1 below shows that this problem does not arise in the present class of spacetimes: the translational KVF, and hence $E$, are both (essentially) uniquely defined.

This section is structured as follows. We begin with the proof of the result outlined above (Proposition 6.1). In Proposition 6.4, we use the results of Section 5 to show that $E$ blows up as $t\to t_M^-$. Proposition 6.2 shows that this is also the case for a certain curvature invariant of the spacetime except in the case $\lmt x_3=1/2$. A different type of pathology arises in this latter case (Proposition 6.3). The remainder of the section gives the results on radial null geodesics required to derive the globally hyperbolic structure of spacetime (Propositions 6.5 and 6.6). We conclude by collecting the relevant results required for the proof of Theorem 1.1.

\begin{proposition}
Consider the spacetime with line element \refb{LE eta}, subject to the Einstein-Scalar Field equations (10) and the regular axis conditions \refb{axis}. Let $\bsy{\xi}_{(\theta)}=\partial/\partial \theta$ be the KVF generating the axial symmetry and let $\bsy{\xi}$ be another KVF of the spacetime that commutes with and is orthogonal to $\bsy{\xi}_{(\theta)}$. Then $\bsy{\xi}=c\bsy{\xi}_{(z)}$ for some constant $c$. Hence, the $C-$energy relative to the axis $\eta=1$ is well-defined and is given by
\beq
\label{C-energy}E=\frac{1}{8}\left(1-2e^{-2\gamma}\frac{x_1}{x_0^2}(1-x_1)\right).
\eeq
\end{proposition}

\begin{proof}
Let $\bsy{\xi}$ be as in the statement of the proposition. Then $\bsy{\xi}$ has components $\xi^{a}=(\alpha,\beta, 0, y)$ where the components depend on $u,v$ and $z$ only. We then have the following seven non-trivial Killing equations for the components of $\bsy{\xi}$:
\begin{subequations}\bea
\beta,_u =0, \qquad \alpha,_v=0, \\
\label{K2}\alpha,_u+2(\bar{\gamma},_u+\bar{\phi},_u)\alpha+\beta,_v+2(\bar{\gamma},_v+\bar{\phi},_v)\beta=0 ,\\
\label{K3}e^{2\bar{\gamma}+4\bar{\phi}}\beta,_z-y,_u=0,\qquad e^{2\bar{\gamma}+4\bar{\phi}}\alpha,_z-y,_v=0,\\
\label{K4}\left(\bar{\phi},_u+\frac{r,_u}{r}\right)\alpha +\left(\bar{\phi},_v+\frac{r,_v}{r}\right)\beta=0,\\
\label{K5}\bar{\phi},_u\alpha+\bar{\phi},_v\beta - y,_z = 0.
\eea\end{subequations} 
We see immediately that $\alpha=\alpha(u,z)$ and $\beta=\beta(v,z)$. Now let 
\beq
P(u,v)= \bar{\phi},_u+\frac{r,_u}{r},\qquad Q(u,v)=\bar{\phi},_v+\frac{r,_v}{r}.
\eeq
Our aim is to show that $\alpha=\beta=0$, which gives $y=y_0$ constant and $\bsy{\xi} = y_0\partial_z$. We consider the three cases which arise from \refb{K4}:
\begin{subequations}
\bea
(i)\quad &P=Q=0, \\
(ii)\quad &P=0, Q\neq 0,\\
(iii) \quad&PQ\neq 0
\eea
\end{subequations}
where the remaining case $Q=0,P\neq0$ is equivalent to case $(ii)$. 
In case $(i)$ we find that $e^{\bar{\phi}}r$ is constant. However, our self-similar solutions have $e^{\bar{\phi}}r=|u|^{1/2}e^{\phi(\eta)}S(\eta)$ which is not constant, so we have a contradiction. \\
In case $(ii)$ we find that $e^{\bar{\phi}}r=q(v)$. Equating this to our self-similar solution we find that $e^{\bar{\phi}}r=|u|^{1/2}e^{\phi(\eta)}S(\eta)=q(v)$ is consistent only if $q=q_0|v|^{1/2}$, which gives $e^{\phi(\eta)}S(\eta)=q_0|\eta|^{1/2}$. However, at the regular axis we have $S(1)=0$ and $\phi(1)$ finite. This sets $q_0=0$, and thus $S=0$, for all $\eta$, which is clearly not the case.\\
Hence, only case $(iii)$ remains. It follows immediately from \refb{K4} that 
\beq
P\alpha+Q\beta=0,\quad
P\alpha,_z+Q\beta,_z=0.
\eeq
Since $PQ\neq0$ we must have 
\beq
\det\left(\begin{array}{cc}\alpha&\beta\\ \alpha,_z&\beta,_z\end{array}\right) = 0,
\eeq
which yields 
\beq
\frac{\alpha,_z(u,z)}{\alpha(u,z)} = \frac{\beta,_z(v,z)}{\beta(v,z)} = h(z). 
\eeq
Integrating then produces 
\beq
\alpha = \bar{\alpha}(u)H(z),\qquad \beta = \bar{\beta}(v)H(z),
\eeq
for some functions $\bar{\alpha},\bar{\beta},H$. Clearly $H=0$ returns $\alpha=\beta=0$. If $H\neq0$ then \refb{K4} may be written as 
\beq
\label{K6}\left(r\bar{\phi},_u+r,_u\right)\bar{\alpha} +\left(r\bar{\phi},_v+r,_v\right)\bar{\beta}=0. 
\eeq
It follows from \refb{SSS} that $r\bar{\phi},_u+r,_u$ and $r\bar{\phi},_v+r,_v$ depend only on $\eta$. It then follows from \refb{K6} that $\bar{\alpha}/\bar{\beta}$ also depends only on $\eta$. It is then straightforward to show that $\bar{\alpha}= \alpha_0|u|^p,\bar{\beta}=\beta_0|v|^p$ for constants $p,\alpha_0,\beta_0$. Substituting these solutions and the self-similar forms for $\bar{\gamma}$ and $\bar{\phi}$ into \refb{K2} and \refb{K6} we find
\begin{subequations}
\bea
\label{exact1}\alpha_0\left(2\eta\phi'+2\eta\gamma'+1-p\right)&=\beta_0\eta^{p-1}(2\eta\gamma'+2\eta\phi'+p),\\
\label{exact2}\left(\eta\phi'S+\eta S'-\frac{S}{2}\right)\alpha_0&=\beta_0\left(\phi'S+S'\right)\eta^p,
\eea
\end{subequations}
assuming $v<0,u<0$ as is the case in region \textbf{I}. 
On the axis, equations \refb{exact1} and \refb{exact2} reduce to
\beq
\alpha_0\left(p-\frac{1}{2}\right)=\beta_0\left(\frac{1}{2}-p\right),\qquad \beta_0=\alpha_0,
\eeq
using \refb{axis}. Assuming $\alpha_0=\beta_0\neq0,$ these have the simultaneous solution $p=1/2$. Solving equation \refb{exact1} we find that
\beq
2\eta\gamma'+2\eta\phi'+\frac{1}{2}=k^2\eta l'/2=0,
\eeq
using \refb{FE a}. This gives constant $l$ which clearly contradicts previous results and so we must have $\alpha_0=\beta_0=0$. It follows that $\alpha=\beta=0$ and $\bsy{\xi}=y_0\partial_z$. Hence, the translational Killing vector $\bsy{\xi}_{(z)}=\partial_z$ is uniquely defined up to a multiplicative constant. It follows that the C-energy relative to the axis $\eta=1$, as given by \refb{C-en}, is well-defined. A straightforward calculation gives the form \refb{C-energy} in terms of the metric components. 
\end{proof}

\begin{proposition}
Let $\cal{T}$ be the scalar curvature invariant $T^{ab}T_{ab}$. Then in the case $\lmt x_3=\pm\infty$ we have $\lmt\cal{T}=+\infty$.
\end{proposition}
\begin{proof}
\noindent In \cite{Part1} it was shown that
\beq
\label{T}\mathcal{T}\ge \frac{k^4e^{-4\gamma-4\phi}}{16v^2}\left(\eta^2 l'(\eta)^2-\frac{1}{4}\right)^2=\frac{k^4e^{-k^2l+t-2c_1}}{16v^2}x_3^2(1-x_3)^2.
\eeq
We first consider the case $\lmt x_3=-\infty$.  In this case, (\ref{x_i}) shows that $l$ is eventually decreasing, so we must have $\lmt e^{-k^2l}>0$. It follows that $\lmt\cal{T}=+\infty$, since $t_M<\infty$.\\
We now consider the case $\lmt x_3=+\infty$. Here, $x_2$ is eventually decreasing and bounded below by zero, so $\lmt x_2$ exists. If $\lmt x_2>0$, then $\lmt e^{-k^2l}>0$ and $\mathcal{T}\to+\infty$ as above. Suppose then that $\lmt x_2=0$. Dividing \refb{x0} by $x_1^2$ and taking the limit $t\rightarrow t_M$ yields
\beq 
\label{lim con} 1-\lmt\left[\frac{e^t}{\sigma'^2}+\frac{k^2}{2}\left(\frac{x_3}{x_1}\right)^2-k^2\left(\frac{x_3}{x_1}\right)\right]=0,
\eeq
where we note that Lemma 5.6 was used to deduce $\lmt 1/x_1=0$. Note that $\sigma''/\sigma = x_1+\delta x_2\to-\infty$ and so $\sigma'$ is decreasing in a neighbourhood of $t_M$. On the other hand we have $\sigma''> \sigma'-b\sigma$ for some positive constant $b$, since $x_2\to 0$. Rewriting this as $\sigma''-\sigma'>-b\sigma$ and integrating yields a lower bound for $\sigma'$. Thus $\sigma_M'=\lmt \sigma'<0$ exists. Letting $ e^{t_M}/\sigma_M'^2=\omega>0$ and $\lmt x_3/x_1 = \ell$ we have
\beq
\ell^2-2\ell= \frac{2}{k^2}(1-\omega)<\frac{2}{k^2},
\eeq
where $\omega>0$. It follows that 
\beq
\ell >1-\sqrt{1+\frac{2}{k^2}}>-\frac{1}{k^2}. 
\eeq
Now consider $Y=e^{-k^2l/4}x_3$, which obeys
\beq
\frac{Y'}{Y}= \frac{1}{2}+\frac{k^2}{8}+\delta\frac{2x_2}{k^2x_3}-x_1\left(1-\frac{1}{2x_3}\right)-\frac{k^2x_3}{4}. 
\eeq
Using $\lmt x_2/x_3=0$ and $\lmt x_3/x_1=\ell$, there must exist $t_0$ sufficiently close to $t_M$ such that $x_3<-x_1/k^2$ and
\beq
\frac{Y'}{Y}>-x_1\left(\frac{3}{4}-\frac{1}{2x_3}\right),
\eeq
on $(t_0,t_M)$. Integrating then shows that $\lmt Y=+\infty$. We observe that the lower bound for $\cal{T}$ seen in \refb{T} behaves like $Y^4$ as $t\rightarrow t_M$, so the proof is complete. 
\end{proof}

\begin{proposition} In the case $\lmt x_3 = 1/2$ the specific length of the cylinders $L$ limits to zero as $t\to t_M$ and the axis located at $t=t_M$ is, therefore, irregular. 
\end{proposition}
\begin{proof}
It follows from \refb{specL} that $L=|u|^{1/2}e^{-\phi}$. From \refb{phi'2} we have $d\phi/dt>-x_1/2$ which may be integrated to show that $\lmt \phi=+\infty$, using Lemma 5.8. This gives $\lmt L=0$ which violates the regular axis conditions \cite{Hayward}. 
\end{proof}

\begin{proposition}
The $C$-energy $E$ satisfies $\lmt E=+\infty$. 
\end{proposition}
\begin{proof}
We consider the behaviour of the term 
\beq
\frac{e^{-2\gamma}x_1^2}{4x_0^2}= \frac{e^{-2\gamma}\sigma'^2}{4e^t},
\eeq
in \refb{C-energy} as $t$ approaches $t_M$. By \refb{C-energy}, this term is a lower bound for $E$ since $x_1\to-\infty$. In the proof of Proposition 6.2, we showed that if $\lmt x_3=+\infty$ then $x_3<-x_1/k^2$ approaching $t_M$. Then in the three cases $\lmt x_3=\pm\infty$ and $\lmt x_3=1/2$ we have $2d\gamma/dt<x_1/2$ approaching $t_M$, via \refb{gam'2}. Integrating then shows that $\lmt \gamma=-\infty$, using Lemma 5.8. If we can show that $\lim_{t\to t_M^-}\sigma'e^{-t/2}\neq 0$, then the conclusion of the proposition holds. We proceed on a case by case basis (sign of $\delta$, limiting value of $x_3$), taking as our starting point the following form of (\ref{x1}):
\beq
\sigma'' = \sigma x_2\left(\frac{x_1}{x_2}+\delta\right). \label{sigpp}
\eeq

In the case $\delta=-1$, we have $\sigma''<\sigma x_1<0$. Thus $\sigma'$ is decreasing and negative (since $\sigma'=\sigma x_1$ and $x_1\to-\infty$), and so $\sigma'\to 0$ cannot arise. 

In the case $\delta=1$ and $\lim_{t\to t_M^-}x_3=+\infty$, we find that (eventually) $\sigma''<0$ - see the proof of Proposition 6.2. If $\delta=1$ and $\lim_{t\to t_M^-} x_3=\frac12$, then (\ref{x2}) shows that $x_2$ is bounded and (\ref{sigpp}) shows that $\sigma''$ is eventually negative. So in both of these cases, we can rule out $\sigma'\to 0$ as we did in the case $\delta=-1$. 

The remaining case is $\delta=1$ and $\lim_{t\to t_M^-}x_3=-\infty$. In this case, we integrate (\ref{bound}) (which holds in the case $\delta=1$) to obtain 
\begin{equation}
\sigma' < \frac{k^2}{2}\sigma x_3 + b e^{t/2} 
\label{sigp-bound}
\end{equation}
for some constant $b$. Now suppose that $\lim_{t\to t_M^-} \sigma'=0$. Then $\sigma x_3$ must be bounded below in the limit $t\to t_M^-$ (remember that $\sigma\to 0$ and $x_3\to-\infty$), and so $\lim_{t\to t_M^-}\sigma'(\sigma x_3)=0$. If we multiply (\ref{x0}) by $\sigma^2$ and take the limit $t\to t_M^-$, we obtain (using $\sigma x_1=\sigma'$)
\begin{eqnarray*}
0 &=& \lim_{t\to t_M^-} (\sigma')^2 \\
&=& \lim_{t\to t_M^-} \left[ e^t +(\frac{k^2}{2}+1)\sigma\sigma'+\frac{k^2}{2}(\sigma x_3)^2 - k^2\sigma'(\sigma x_3)+2\delta\sigma^2x_2\right]\\
&\geq& e^{t_M},
\end{eqnarray*}
yielding a contradiction. Therefore $\lim_{t\to t_M^-}\sigma'\neq0$. 

So in all cases, $\lim_{t\to t_M^-}\sigma'$ cannot be zero, and $E$ diverges as claimed.
\end{proof}

\noindent\textbf{Comment 6.1} The three preceding results establish that, in one way or other, the spacetimes corresponding to the semi-global solutions of Section 5 are singular at $t=t_M$: in fact the blow-up of the well-defined C-energy $E$ shows that the spacetimes are $C^1-$inextendible across this surface. We now derive results relating to radial null geodesics of these spacetimes, and thereby establish their global hyperbolicity. We note that the term `radial' here means that $z$ and $\theta$ are constant along the geodesic. 

\begin{proposition}
All outgoing radial null geodesics terminate in the future at $t=t_M$ in finite affine parameter time.
\end{proposition}
\begin{proof} Outgoing null geodesics are lines of constant $u$ (see Figure 2). For outgoing radial null rays we have $\dot{u}=\dot{\theta}=\dot{z}=0$ along the geodesic where the overdot represents a derivative with respect to an affine parameter $\mu$, which is chosen such that $\dot{v}>0$ and $\mu(\eta=0)=0$. The equation governing these geodesics then reduces to 
\beq
\ddot{v}+(2\bar{\gamma},_v+2\bar{\phi},_v)\dot{v}^2=0.
\eeq
Dividing by $\dot{v}$, integrating and taking the exponential yields
\beq
e^{2\bar{\gamma}+2\bar{\phi}}\dot{v}=|u_0|^{-1}e^{2\gamma+2\phi}\dot{v}=|u_0|^{-1}e^{k^2l/2-t/2+c_1}\dot{v}=C>0,
\eeq
where $u_0$ is constant. This equation is valid along lines $u=u_0$ and we remind the reader that $u<0$ everywhere to the past of the singularity. On these lines we have $v=u_0\eta=-u_0e^t$ and thus $dv=-u_0e^tdt$. Integrating the above equation over $\mu'\in(0,\mu)$ then gives
\beq
\label{geo2}\int_{0}^v |u_0|^{-1}e^{k^2l/2-t/2+c_1}dv'=\int_{-\infty}^te^{k^2l/2+t/2+c_1}dt'=C\mu.
\eeq
Given that $\lmm l'=-1/2$ we have $e^{k^2l/2+t/2}\sim e^{\LL t/2}$ as $t\rightarrow-\infty$ and it is straightforward to show that there exists $t_0\in(-\infty,t_M)$ such that the portion of the integral in \refb{geo2} over the interval $(-\infty,t_0)$ is finite. In fact, this is true of any $t_0$ which is bounded away from $t_M$. 
Thus the nature of $\mu$ - that is, the question of whether it is finite or infinite - is completely determined by the limiting behaviour as $t\to t_M^-$ of
\begin{equation} \tilde{\mu}(t) = \int_{t_0}^t e^{k^2l/2} dt' = |V_0|^{\frac{k^2}{2|\lambda|}}\int_{t_0}^t (x_2(t'))^{-\frac{k^2}{2|\lambda|}}dt'.\label{mu-tilde} \end{equation}
We complete the proof by showing that $\tilde{\mu}$ is finite in each of the three cases $\lim_{t\to t_M^-}x_3=-\infty,\frac12,+\infty$. It is convenient for this to recall (\ref{x2}):
\begin{equation}
x_2'(t) = |\lambda|(\frac12-x_3(t))x_2(t),\label{x2-new}
\end{equation}
or equivalently, 
\begin{equation} x_2(t)=x_2(t_0)\exp\{\int_{t_0}^t|\lambda|(\frac12-x_3(t'))dt'\},\label{x2-new-int}
\end{equation}
It follows immediately that if $\lim_{t\to t_M^-}x_3=-\infty$ or $\frac12$, then $x_2$ has a positive lower bound, and so $\tilde{\mu}$ is bounded above. This proves the proposition in these two cases. 

In the case where $x_3\to +\infty$, (\ref{x2-new}) shows that $x_2$ (which is non-negative by definition) is eventually decreasing, and therefore has a non-negative limit. (In addition, $x_2$ is bounded above.) If this limit is positive, we repeat the argument used in the previous lines. 

Thus it remains only to deal with the case where $x_3\to +\infty$ and $x_2\to 0$ as $t\to t_M^-$.  Recall also that we must have $x_1\to-\infty$ in the limit (see Propositions 5.1 and 5.2). Let $u_1=x_1^{-1}$. Then $u_1\to 0$ as $t\to t_M^-$, and we can write (\ref{x1}) as 
\begin{equation}
u_1'=1-u_1-\delta x_2 u_1^2.\label{u1p}
\end{equation}
Thus $u_1'\sim 1$ as $t\to t_M^-$, and we can integrate to obtain
\begin{equation} u_1 \sim t-t_M, \qquad t\to t_M^-, \label{u1-asymp} \end{equation}
and so 
\begin{equation} x_1 \sim (t-t_M)^{-1}, \qquad t\to t_M^-. \label{x1-asymp} \end{equation}
In the proof of Proposition 6.2, we deduced in the case where $x_3\to + \infty$ and $x_2\to 0$ (cf. the paragraph containing (\ref{lim con})) that for all $t$ sufficiently close to $t_M$,
\begin{equation} x_3 < -\frac{x_1}{k^2}.\label{x3-bound-x1}\end{equation} Choosing $t_0$ sufficently close to $t_M$ in (\ref{x2-new-int}) then gives
\begin{eqnarray}
x_2(t) &>& x_{2,L}:=x_2(t_0)\exp\{ \int_{t_0}^t |\lambda|(\frac12 + \frac{x_1}{k^2}) dt'\}\nonumber\\
&\sim& x_2(t_0)|t-t_M|^{|\lambda|/k^2},\qquad t\to t_M^- \label{x2-asymp}
\end{eqnarray}
where we have used (\ref{x1-asymp}). Then
\begin{eqnarray}
e^{k^2l/2} &=& \left(\frac{x_2}{|V_0|}\right)^{-\frac{k^2}{2|\lambda|}}\nonumber\\
&<& \left(\frac{x_{2,L}}{|V_0|}\right)^{-\frac{k^2}{2|\lambda|}}\nonumber\\
&\sim& C|t-t_M|^{-1/2},\qquad t\to t_M^-
\end{eqnarray}
for some positive constant $c$. Since this term is integrable on any interval of the form $(t_0,t_M)$, we see from (\ref{mu-tilde}) that $\tilde{\mu}(t)$ is finite in the limit $t\to t_M^-$. This completes the proof.
\end{proof}

\begin{proposition} 
Ingoing radial null geodesics have infinite affine length to the past. For $v>0$ they have finite affine length to the future. 
\end{proposition}
\begin{proof}
For ingoing geodesics we have $v=v_0,\dot{\theta}=\dot{z}=0$ and $u=v/\eta=-v_0e^{-t}, du=-udt$. Solutions to the geodesic equation are then given by
\beq
\label{geo length}e^{2\gamma+2\phi}\frac{dt}{d\mu}=e^{k^2l/2-t/2+c_1}\frac{dt}{d\mu}=\tilde{C}.
\eeq
To determine whether the spacetime has a past null infinity we look for $\lim_{u\rightarrow-\infty}\mu$. In terms of $t$, this is given by $\lmm\mu$. Integrating over $(t,t_0)$ and taking this limit we find
\beq
\tilde{C}\lmm(\mu_0-\mu)=\lmm\int_t^{t_0}e^{k^2l/2-t/2+c_1}dt'.
\eeq
Given that $\lmm l'=-1/2$, we clearly have $\lmm\mu=-\infty$. To calculate the future affine length along the geodesics from a fixed $u_0$ we integrate \refb{geo length} over $(t_0,t_M)$. It then follows from the proof of Proposition 6.5 that this length is finite (this argument applies when $v>0$; these ingoing geodesics extend to $t=t_M$). For completeness we now examine the behaviour of the null geodesic along $\cal{N}_-$. Our current coordinate system is not suited to the task since some of the metric functions blow up there. Specifically, we have seen that $e^{2\gamma+2\phi}\sim e^{-(k^2/4+1/2)t}=(-\eta)^{-(k^2/4+1/2)}$ in the limit as $t\rightarrow-\infty,\eta\rightarrow0^-$. We define 
\beq
\xi(\eta)=\int_0^\eta e^{2\gamma(\eta')+2\phi(\eta')}d\eta',\qquad \beta = \frac{k^2}{4}+\frac{1}{2}.
\eeq
Since $e^{2\gamma+2\phi}\sim(-\eta)^{-\beta}$ and $0<\beta<1$ for $k^2<2$, it is straightforward to show that $\xi\sim\eta^{1-\beta},\xi(0)=0$. 
Writing the line element in terms of $\xi$ we have 
\bea
ds^2&=2e^{2{\gamma}+2{\phi}}(dud\eta+u^{-1}\eta\, du^2)+e^{2\bar{\phi}}r^2d\theta^2+e^{-2\bar{\phi}}dz^2\\
\nonumber&= 2dud\xi + 2e^{2\gamma(\eta(\xi))+2\phi(\eta(\xi))}u^{-1}\eta(\xi)\, du^2+e^{2\bar{\phi}}r^2d\theta^2+e^{-2\bar{\phi}}dz^2.
\eea
To derive the geodesic equation we consider the Langrangian $\mathscr{L}$ which simplifies to 
\beq
\Lag =2\dot{u}\dot{\xi}+2u^{-1}\eta e^{2\gamma+2\phi} \dot{u}^2,
\eeq
for radial null geodesics. We then have
\beq
\frac{d}{d\mu}\frac{\partial\Lag}{\partial\dot{\xi}}-\frac{\partial\Lag}{\partial\xi}=2\ddot{u}-2u^{-1}\frac{d}{d\xi}\left(\eta e^{2\gamma+2\phi}\right)\dot{u}^2=0.
\eeq
Using $d\xi=e^{2\gamma+2\phi}d\eta$ and the derivative of \refb{FE a} we have 
\bea
\frac{d}{d\xi}\left(\eta e^{2\gamma+2\phi}\right)=1+\eta(2\gamma'(\eta)+2\phi'(\eta))=\frac{k^2\eta l'}{2}+\frac{1}{2}.
\eea
Given that $\eta l'=-1/2$ everywhere on $\cal{N}_-$, the geodesic equation reduces to 
\beq
\frac{\ddot{u}}{\dot{u}}-\frac{\LL\dot{u}}{2u}=0,
\eeq
which may be integrated to give 
\beq
|\dot{u}|=C|u|^{\LL/2}, 
\eeq
for some $C>0$. Choosing the affine parameter such that $\mu(u=0)=0$ and $\dot{u}<0$ and integrating over $(u,0)$ we find $|u|^{1-\LL/2}=\tilde{C}\mu$ where $\tilde{C}=(1-\LL/2)C$. Hence, we find that $\lim_{u\rightarrow-\infty}\mu=+\infty$. 
\end{proof}

\noindent\textbf{Comment 6.2} The results of this section show that the structure of the spacetimes are as shown in Figure 3. The point ${\cal P}$ corresponds to the limit $(v,u)\to(\infty,-\infty)$ subject to $\eta=\eta_M$. Any spacelike surface extending from the axis to the point $\cal{P}$ such as the one depicted by the dashed line represents a Cauchy surface of the spacetime. The spacetimes are, therefore, globally hyperbolic. 

\begin{center}\begin{figure}[h]
\includegraphics[scale=0.5]{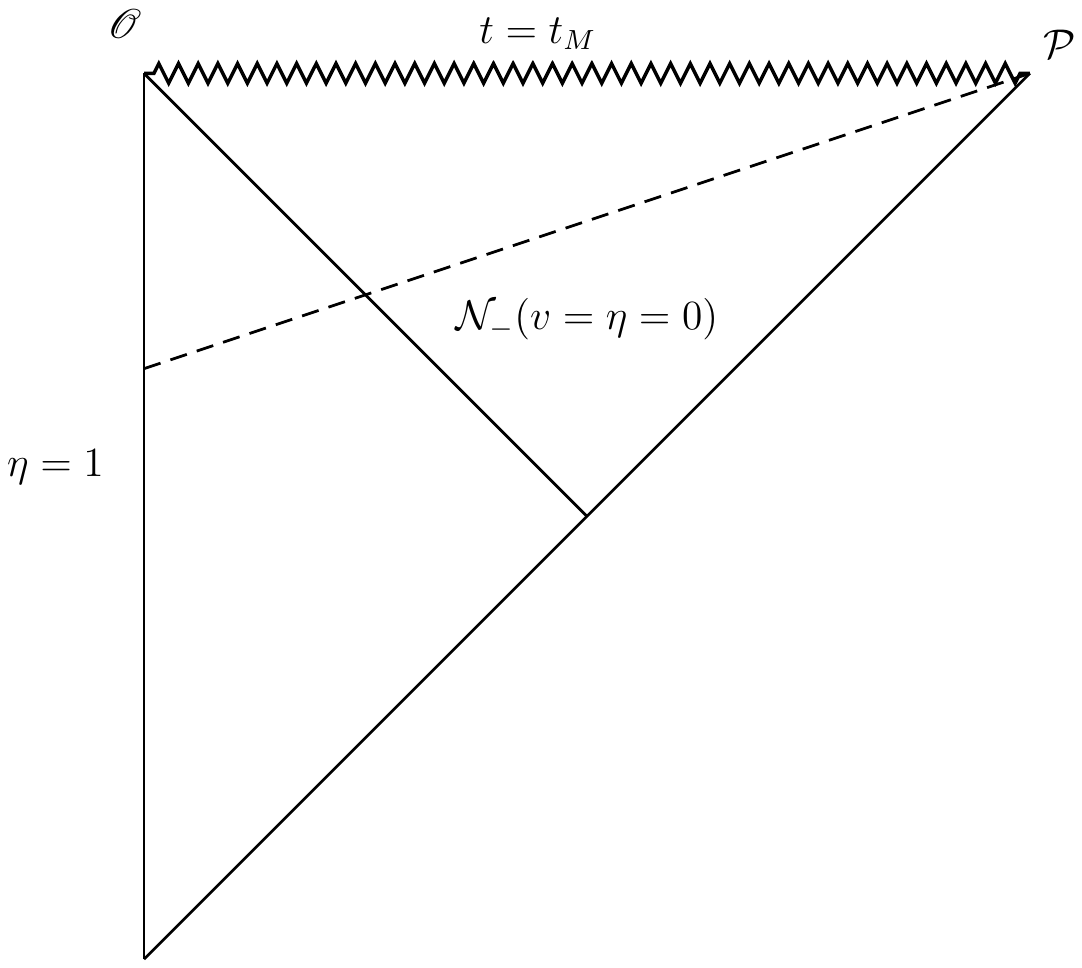}
\vskip-150pt
\caption{Global structure of the spacetimes with $k^2<2$. The dashed line represents a Cauchy surface for the spacetime.}
\end{figure}
\end{center}

\noindent\textbf{Proof of Theorem 1.1}
\begin{proof}
We collect the results above to show that this class of spacetimes is globally hyperbolic and $C^1$-inextendible. The latter follows from Propositions 6.1 and 6.4 which show that the invariant $E$, which depends only on the metric and its first derivatives, blows up as $t\to t_M$. Global hyperbolicity follows from Propositions 6.5 and 6.6, which yield the conformal diagram of Figure 3. Regularity of the axis ensures that ingoing causal geodesics meeting the axis make a smooth transition to outgoing causal geodesics.
\end{proof}
\section{Acknowledgments}
We thank Filipe Mena for useful comments on the results of this paper. We are also very grateful to the referees for their invaluable feedback, particularly in highlighting errors in a previous version of this paper. This project was funded by the Irish Research Council for Science, Engineering and Technology, grant number P07650.

\end{document}